\documentclass[12pt,journal,draftclsnofoot,onecolumn,letterpaper]{IEEEtran}
\usepackage{mathrsfs}
\usepackage{amssymb}
\usepackage{amsmath}
\usepackage{amsmath,bm}
\usepackage{amsthm}
\usepackage{graphicx}
\usepackage{subfigure}
\usepackage{cite}
\usepackage{enumerate}
\usepackage{color, soul}
\usepackage{verbatim}
\usepackage{amsfonts}
\usepackage{array}
\usepackage{bm}
\usepackage{subfigure}
\usepackage{diagbox}
\usepackage{algorithmic}
\usepackage{algorithm}

\begin{document}

\newtheorem{definition}{\bf ~~Definition}
\newtheorem{observation}{\bf ~~Observation}
\newtheorem{lemma}{\bf ~~Lemma}
\newtheorem{proposition}{\bf ~~Proposition}
\newtheorem{remark}{\bf ~~Remark}
\renewcommand{\algorithmicrequire}{ \textbf{Input:}} 
\renewcommand{\algorithmicensure}{ \textbf{Output:}} 
\newcommand*{\TitleFont}{%
      \fontsize{18}{18}%
      \selectfont}
\newcommand*{\AuthorFont}{%
      \fontsize{12}{12}%
      \selectfont}
\title{\TitleFont UAV-to-Device Underlay Communications: Age of Information Minimization by Multi-agent Deep Reinforcement Learning}

\author{
\IEEEauthorblockN{
{\AuthorFont Fanyi Wu}, \IEEEmembership{\AuthorFont Student Member, IEEE},
{\AuthorFont Hongliang Zhang}, \IEEEmembership{\AuthorFont Member, IEEE},
{\AuthorFont Jianjun Wu},\\
{\AuthorFont Lingyang Song}, \IEEEmembership{\AuthorFont Fellow, IEEE},
{\AuthorFont Zhu Han}, \IEEEmembership{\AuthorFont Fellow, IEEE},
{\AuthorFont and H. Vincent Poor}, \IEEEmembership{\AuthorFont Fellow, IEEE}}\\

\vspace{-1.0cm}

\thanks{F.~Wu, J.~Wu, and L.~Song are with Department of Electronics Engineering, Peking University, Beijing 100871, China (email: fanyi.wu@pku.edu.cn, just@pku.edu.cn, lingyang.song@pku.edu.cn).}
\thanks{H.~Zhang is with Department of Electronics Engineering, Peking University, Beijing 100871, China, and also with Department of Electrical and Computer Engineering, University of Houston, Houston, TX 77004, USA (email: hongliang.zhang92@gmail.com).}
\thanks{Z.~Han is with Department of Electrical and Computer Engineering, University of Houston, Houston, TX 77004, USA, and also with Department of Computer Science and Engineering, Kyung Hee University, Seoul 02447, South Korea (email: hanzhu22@gmail.com).}
\thanks{H. V. Poor is with Department of Electrical Engineering, Princeton University, Princeton, NJ 08544, USA (e-mail: poor@princeton.edu).}

}

\maketitle

\vspace{-1.0cm}

\begin{abstract}
In recent years, unmanned aerial vehicles~(UAVs) have found numerous sensing applications, which are expected to add billions of dollars to the world economy in the next decade. To further improve the Quality-of-Service~(QoS) in such applications, the 3rd Generation Partnership Project~(3GPP) has considered the adoption of terrestrial cellular networks to support UAV sensing services, also known as the cellular Internet of UAVs. In this paper, we consider a cellular Internet of UAVs, where the sensory data can be transmitted either to the base station~(BS) via cellular links, or to the mobile devices by underlay UAV-to-Device (U2D) communications. To evaluate the freshness of the sensory data, the age of information~(AoI) is adopted, in which a lower AoI implies fresher data. Since UAVs' AoIs are determined by their trajectories during sensing and transmission, we investigate the AoI minimization problem for UAVs by designing their trajectories. This problem is a Markov decision problem~(MDP) with an infinite state-action space, and thus we utilize multi-agent deep reinforcement learning~(DRL) to approximate the state-action space. Then, we propose a multi-UAV trajectory design algorithm to solve this problem. Simulation results show that our proposed algorithm can achieve a lower AoI than the greedy algorithm and the policy gradient algorithm.
\end{abstract}

\begin{IEEEkeywords}
UAV-to-Device communication, cellular Internet of UAVs, age of information, multi-agent deep reinforcement learning
\end{IEEEkeywords}

\newpage

\section{Introduction}%
\label{Introduction}

As an emerging facility with high mobility and low operational cost~\cite{HJJYG-2018,HLZ-2020}, the unmanned aerial vehicle~(UAV) has been applied to provide sensing services in a wide range of areas, including road traffic monitoring~\cite{KGMK-2013}, forest fire surveillance~\cite{CZY-2017}, and industrial facility inspection~\cite{JMJSCR-2013}. Following the announcement of the Federal Aviation Administration~(FAA) in 2019, the UAV is anticipated to generate tens of billion dollars to the world economy~\cite{FAA}. In the current unmanned aircraft systems~(UASs), UAVs transmit their sensory data to terrestrial mobile devices over the unlicensed spectrum~\cite{RJW-2013}, in which the interference from surrounding terminals is uncontrollable due to the opportunistic unlicensed spectrum access. To ensure the Quality-of-Services~(QoS) in such applications, the 3rd Generation Partnership Project~(3GPP) has considered the adoption of terrestrial cellular networks to support UAV sensing services~\cite{3GPP_TR_36_777}, which is also referred to as the cellular Internet of UAVs~\cite{HLZH-2019}. In the cellular Internet of UAVs, UAVs can transmit their sensory data to the base stations~(BSs) and corresponding mobile devices for different applications over the licensed spectrum, which are respectively known as the UAV-to-Network~(U2N) and the UAV-to-Device~(U2D) communications.

In some sensing applications, the conditions of sensing targets change rapidly, and thus UAVs are required to perform sensing and transmission continuously to keep their sensory data up-to-date. To evaluate the freshness of sensory data, the age of information~(AoI) is adopted as one metric~\cite{SRM-2012}. Specifically, the AoI of a UAV is defined as the elapsed time after the latest successful transmission of the valid sensory data. When a UAV has a high AoI, its latest sensory data may be inconsistent with the current condition of its target. Consequently, the UAV has the incentive to minimize its AoI. Furthermore, due to the limited sensing range of onboard sensors, a UAV tends to fly close to its target for successful sensing. However, this UAV may suffer low throughput as it moves far from the BS or its mobile device, which leads to long duration of sensory data  transmissions. Therefore, the AoI of a UAV is jointly determined by its trajectory during sensing and transmission. With the aim to minimize the AoI, it is necessary to design the trajectory of the UAV.

In this paper, we consider an orthogonal frequency division multiple access~(OFDMA) cellular Internet of UAVs, in which UAVs can transmit their sensory data to the BS through U2N links, or directly to their mobile devices through U2D links. To make full use of the spectrum resource, U2D communications work as an underlay to U2N ones. Since UAVs' AoIs are jointly influenced by their trajectories during sensing and transmission, we aim to minimize UAVs' AoIs by designing their trajectories. Furthermore, in the system, UAV sensing and transmission are coupled with each other. Besides, due to the underlay property, different UAVs' trajectories may influence on each other. Therefore, it is challenging to investigate the AoI minimization problem in our system.

To tackle with this challenge, we design a joint sensing and transmission protocol to schedule UAVs performing sensing tasks.  Based on the proposed protocol, UAV sensing and transmission can be formulated as the state transitions in the Markov chains. Consequently, the AoI minimization problem can be regarded as a Markov decision problem~(MDP)~\cite{RA-1998}. Since UAVs' states and actions in this MDP are continuous-valued, the state-space is infinite. Besides, as the interference among UAVs is not observable, traditional model-based methods can not be leveraged to tackle with this MDP. Therefore, we adopt multi-agent deep reinforcement learning~(DRL)~\cite{KMMA-2017,JHLZH-publish} to solve this problem, and propose a multi-UAV trajectory design algorithm based on deep deterministic policy gradient~(DDPG)~\cite{TJANTYDD-2016} method to optimize the policies for UAVs.

In the literature, several works have investigated the UAV communications in the cellular Internet of UAVs.
Specifically, authors in~\cite{SHBL2-2018,AK-2018} focused on the U2N communications. In~\cite{SHBL2-2018}, the authors  investigated a cellular Internet of UAVs consisting of single UAV, and maximized the energy efficiency in the network by jointly optimizing the UAV's trajectory and transmission power. In~\cite{AK-2018}, the authors studied a cellular network including one UAV performing sensing tasks, and modeled statistical behavior of the channel from the BS to the UAV based on extensive experimental data measurements.
Moreover, the authors in~\cite{SHBL3-2018,ZJLKG-2018} studied the UAV-to-UAV~(U2U) communications, in which multiple UAVs
can communicate with each other directly. In~\cite{SHBL3-2018}, the authors jointly optimized subchannel allocation and flying speed for a cellular Internet of UAVs, where UAVs can not only communicate with the BS through the U2N links, but also build U2U links with other UAVs. In~\cite{ZJLKG-2018}, the authors evaluated the reliability and polling delay for UAV swarms in a cellular network where these UAVs execute the sensing tasks cooperatively.
However, as an important practical scenario in the cellular Internet of UAVs, the direct communications between UAVs and mobile devices, namely the U2D communications, are lack of considerations in the current works. Therefore, we propose to enable U2D communications into the cellular Internet of UAVs in this paper.

The main contributions of this paper can be summarized as follows:
\begin{itemize}
  \item We propose the underlay U2D communications in the cellular Internet of UAVs to improve the QoS for UAV sensing services, and design a joint sensing and transmission protocol to enable U2D communications.
  \item We investigate the AoI minimization problem for UAVs in the cellular Internet of UAVs using multi-agent DRL, and then propose a DDPG-based multi-UAV trajectory design algorithm to solve this problem.
  \item Simulation results show that our proposed algorithm can achieve a lower AoI in the cellular Internet of UAVs than the greedy algorithm and the policy gradient algorithm.
\end{itemize}

The rest of our paper is organized as follows. In Section~\ref{System Model}, we describe the model of
cellular Internet of UAVs, in which U2D communications work as an underlay of U2N ones. In Section~\ref{Joint Sensing and Transmission Protocol}, we design a joint sensing and transmission protocol to schedule UAVs performing sensing tasks. Then, we investigate the AoI minimization problem for UAVs in the system, and reformulate it under multi-agent reinforcement learning~(RL) framework in Section~\ref{AoI Minimization Problem Formulation}. After that, we adopt multi-agent DRL to analyze this problem, and propose a DDPG-based multi-UAV trajectory design algorithm to solve it in Section~\ref{Algorithm Design by Multi-agent DRL}. In Section~\ref{System Performance Analysis}, we analyze the convergence and the complexity of our proposed algorithm, and remark some properties on UAVs' AoIs. Simulation results are presented in Section~\ref{Simulation Result}. Finally, Section~\ref{Conclusion} concludes this paper.

\section{System Model}
\label{System Model}

\begin{figure}[!t]
\centering
\includegraphics[width=6.0in]{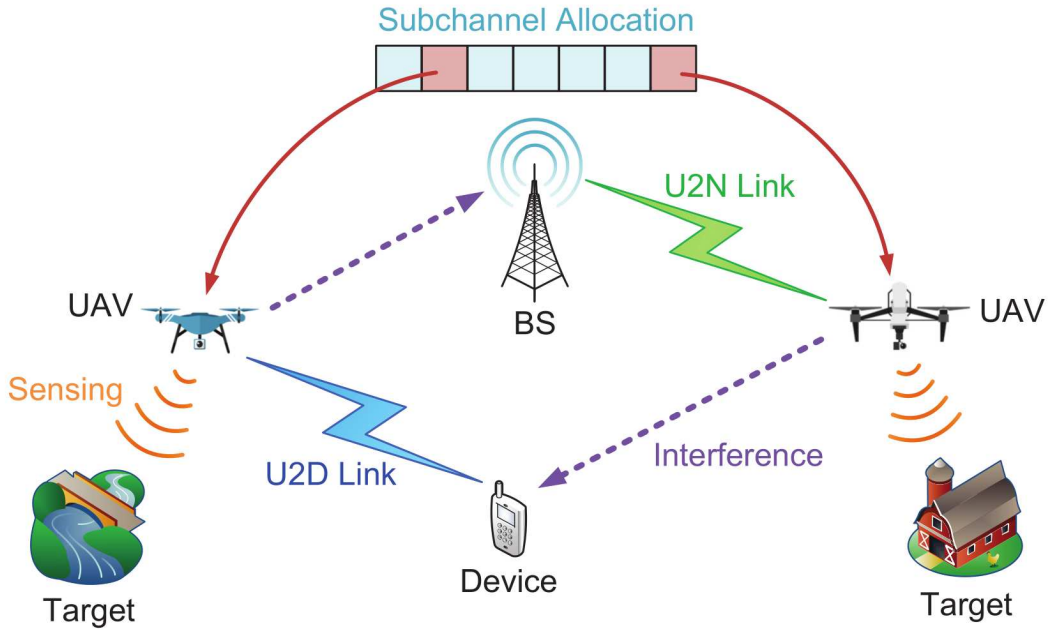}
\vspace{-10mm}
\caption{System model for the cellular Internet of UAVs.}
\vspace{-7mm}
\label{system_model}
\end{figure}

As illustrated in Fig.~\ref{system_model}, we consider an OFDMA cellular Internet of UAVs with one BS and multiple UAVs. Each UAV is required to execute its sensing task involving one target, and then transmits the sensory data to the BS or corresponding mobile device for further processing. We assume that the system includes $M$ UAVs which transmit their data to the BS, denoted by $\mathcal {M} = \{1,2,...,M\}$. There also exist $N$ UAVs which send the data to their mobile devices, denoted by $\mathcal {N} = \{M+1,M+2,...,M+N\}$. Here, $M$ and $N$ are constants which will be determined by the UAV sensing applications. To support the sensing applications, there exist two transmission modes:
\begin{itemize}
  \item \textbf{U2N mode}: The UAV transmits the sensory data to the BS via cellular communications;
  \item \textbf{U2D mode}: The UAV transmits the sensory data to its corresponding mobile device directly.
\end{itemize}
Moreover, the UAVs perform sensing and transmission in a synchronized manner, where the minimum time unit in the system is defined as \emph{frame}. As such, both sensing and transmission of a UAV can be characterized by frames.

We assume that the system owns $K$ orthogonal subchannels to support the data transmissions of UAVs, which are denoted by $\mathcal {K} = \{1,2,...,K\}$. To guarantee the QoS, the BS assigns one exclusive subchannel to each U2N link, and thus the severe mutual interference among U2N transmissions can be avoided. For the sake of fairness, each subchannel can only be occupied by one UAV in the U2N mode at a time. Besides, to make full use of the spectrum resource,  U2D communications work as an underlay to U2N ones, i.e., U2D links are allowed to share subchannels with U2N links. We assume that the BS only allocates one subchannel to each UAV in the U2D mode at a time, while each subchannel can be occupied by at most $V_s$ UAVs for U2D transmissions. To specify the subchannel allocation result, we define a $(M+N)\times K$ binary matrix $\Phi= [\phi_{i,k}]$, $(i \in \mathcal {M} \cup \mathcal {N}, k \in \mathcal {K})$, for all UAVs in the system, in which the indicator $\phi_{i,k}=1$  if the $k$-th subchannel is allocated to the $i$-th UAV; otherwise, $\phi_{i,k}=0$.

We describe the locations of the BS, UAVs, targets, and mobile devices by 3D cartesian coordinates. To be specific, the BS is located at $\textbf{\emph{x}}_0 = (0, 0, H_0)$, in which $H_0$ denotes its height. In addition, the $i$-th UAV $( i \in \mathcal {M} \cup\mathcal {N} )$ is located at $\textbf{\emph{x}}_i = (x_i, y_i, h_i)$, whose target has the coordinate $\textbf{\emph{x}}^t_i = (x^t_i, y^t_i, 0)$. Moreover,  the $j$-th mobile device\footnote{The $j$-th mobile device refers to the corresponding mobile device of the $j$-th UAV.}~$( j \in \mathcal {N} )$ is located at $\textbf{\emph{x}}^d_j = (x^d_j, y^d_j, 0)$.

\subsection{UAV Sensing}

Each UAV is equipped with onboard sensors to sense its target. However, due to the mechanical limitations of onboard sensors, the sensing is not always successful. In this paper, we adopt the probabilistic sensing model in~\cite{JHL-2019,SHBL2-2018} to evaluate the sensing qualities of UAVs. More explicitly, the successful sensing probability~(SSP) for a UAV can be expressed as an exponential function of the distance between the UAV and its target. When the $i$-th UAV $( i \in \mathcal {M} \cup\mathcal {N} )$ sense its target for one frame, its SSP can be expressed by
\begin{equation}\label{SSP}
 \mathcal{P}^{ss}_i (\textbf{\emph{x}}_i )= \left\{ {\begin{array}{*{20}{cc}}
e^{-\lambda t_f  d^t_i} , &\quad d^t_i  \sin\varphi \leq r_{s,i},\\
0,& \quad \textmd{otherwise},
\end{array}} \right.
\end{equation}
where $\lambda$ is the \emph{sensing factor} evaluating the sensing performance, $t_f$ is the duration of a frame, and $d^t_i = \|\textbf{\emph{x}}_i - \textbf{\emph{x}}^t_i\|_2$ is the distance from the UAV to its target. Here, $\|\cdot\|_2$ denotes the Euclidean distance. Moreover, $\varphi$ is the \emph{maximum sensing angle} for UAV onboard sensors, and $r_{s,i} = h_i \tan\varphi$ denotes the \emph{sensing range}, namely the maximum horizon distance between the UAV and its target satisfying $\mathcal{P}^{ss}_i>0$.

When a UAV successfully senses its target, the sensory data is defined as \emph{valid}. Due to the limited onboard computational capability, a UAV cannot figure out whether its sensory data is valid or not by itself. However, after receiving the sensory data from the UAV, the BS or the mobile device can judge whether the sensing is successful or not. As such, the sensing quality of a UAV can still be evaluated based on~(\ref{SSP}).

\subsection{UAV Transmission}
\label{Transmission Model}

As UAVs fly at a high altitude, the line-of-sight~(LoS) components usually exist in the data transmissions of UAVs. Therefore, the channel characteristics of air-to-ground communications are different from that in traditional terrestrial communications. In this paper, we adopt the channel model in~\cite{3GPP_TR_38_901, 3GPP_TR_36_777} to evaluate the data transmissions in the U2N and the U2D modes.

\subsubsection{U2N Mode}
Since the spectrum resource is orthogonally utilized for U2N communications, the mutual interference among U2N links can be avoided. However, due to the underlay property of U2D communications, a UAV in the U2N mode may still be interfered by the co-channel U2D links. Specifically, for the $i$-th UAV in the U2N mode~$( i \in \mathcal {M} )$ over the $k$-th subchannel, the received signal to interference plus noise ratio~(SINR) at the BS can be expressed as
\begin{equation}\label{U2N_SNR}
\gamma_{i,k} = \frac{ \phi_{i,k} P^u g_i \zeta_i }{N_0 + \sum\limits_{j \in \mathcal {N}} \phi_{j,k}  P^u g_j \zeta_j }.
\end{equation}
Here, $\phi_{i,k}$ and $\phi_{j,k}$ are the subchannel allocation indicators, $P^u$ denotes the transmit power of UAVs, and $N_0$ denotes the power of noise. Besides, $g_i = 10^{-\mathcal {L}_i/10}$ and $g_j = 10^{-\mathcal {L}_j/10}$ denote the channel gains from the $i$-th UAV and the $j$-th UAV to the BS, in which $\mathcal {L}_i$ and $\mathcal {L}_j$ denote the air-to-ground path losses, respectively. Moreover, $\zeta_i$ and $\zeta_j$ denote the small-scale fading coefficients.

To calculate the air-to-ground path losses and the small-scale fading coefficients, both the LoS and the none LoS~(NLoS) components should be considered. For the $i$-th UAV $( i \in \mathcal {M} \cup \mathcal {N} )$, the probability of the LoS component can be calculated by
\begin{equation}
\mathcal{P}^{LoS}_i=
\left\{
\begin{aligned}
{d_{c}}/{d_{2D}} + \left( 1 - {d_{c}}/{d_{2D}}\right)
e^{-{d_{2D}}/{p_0}},&\quad{d_{2D} \geq d_c},\\
1,\qquad\qquad\qquad&\quad\textmd{otherwise},\\
\end{aligned}
\right.
\end{equation}
in which $p_0 = 233.98 \textmd{lg}(h_i) - 0.95$, and $d_c = \textmd{max}\{ 294.05\textmd{lg}(h_i)-432.94,18\}$. Besides, $d_{2D}=\sqrt{(x_i)^2 + (y_i)^2}$ implies the 2D distance from the UAV to the BS. Then, the probability of the NLoS component for the $i$-th UAV can be given by $\mathcal{P}^{NLoS}_i = 1 - \mathcal{P}^{LoS}_i$. According to~\cite{3GPP_TR_36_777}, we can calculate the LoS and the NLoS path losses for the $i$-th UAV, denoted by $\mathcal {L}^{LoS}_i$ and $\mathcal {L}^{NLoS}_i$, respectively. Moreover, the LoS and the NLoS small-scale fading coefficients for the $i$-th UAV, denoted by $\zeta^{LoS}_i$ and $\zeta^{NLoS}_i$, obey Rice distribution and Rayleigh distribution, accordingly. More details on calculating the path loss and the small-scale fading coefficient are referred in~\cite{3GPP_TR_36_777}.

Therefore, based on the received SINR, the throughput of the $i$-th UAV~$( i \in \mathcal {M} )$ over the $k$-th subchannel is given by $R_{i,k} = \log_2\left(1+\gamma_{i,k}\right)$. Considering the QoS of data transmission, we assume that a transmission is successful only when the throughput exceeds a given threshold, denoted by $R_{th}$. Then, we can calculate the successful transmission probability~(STP) for U2N transmissions by the following proposition.

\begin{proposition}\label{proposition1}
The STP of the $i$-th UAV in the U2N mode~$( i \in \mathcal {M} )$ over the $k$-th subchannel is given by
\begin{equation}
 \mathcal{P}^{st}_{i,k}(R_{th})  =  \int_0^\infty \left\{  \mathcal{P}^{LoS}_i \cdot [1-F_{ri}(x)]+ \mathcal{P}^{NLoS}_i \cdot [1-F_{ra}(x)] \right\}  f_\chi(x) dx,
\end{equation}
where $F_{ri}(x) = 1 - Q_1(\sqrt{2K_{ri}},x \sqrt{2(K_{ri}+1)})$ and $F_{ra}(x) = 1 - e^{-x^2/2}$.
Besides, $f_\chi(x) = g(x-A_i)$, with $ g(x) = g_{\mathcal {W}[1]}(x) \ast \ldots \ast g_{\mathcal {W} [N_w]}(x)$, $A_i = \frac{N_0\kappa }{P^u g_i }$, and $\kappa = 2^{R_{th}}-1$. Here, $\mathcal {W} = \left\{ j \in \mathcal {N} | \phi_{j,k}=1  \right\}$ whose size is $N_w$, $\mathcal {W}[w]$ is the $w$-th element in $\mathcal {W}$,  $g_j(x)= ({1}/{B_i})  f_{\chi_j}({x}/{B_i})$, $B_i = \frac{ \kappa }{g_i}$, $f_{\chi_j}(y)= \mathcal{P}^{LoS}_j \cdot [({1}/{g_j}) f_{ri}(y/{g_j})] + \mathcal{P}^{NLoS}_j \cdot [({1}/{g_j}) f_{ra}(y/{g_j})]$,  $f_{ra}(y)=y e^{-y^2/2}$, and $f_{ri}(y)= 2\left(K_{ri}+1\right)y e^{ -\left(K_{ri}+1\right)y^2-K_{ri} } \cdot I_0\left(2\sqrt{(K_{ri}+1)K_{ri}} y\right)$.

\end{proposition}
\begin{proof}
See Appendix~\ref{proof1}.
\end{proof}

Based on Proposition~\ref{proposition1}, we can further calculate the expected throughput for U2N transmissions as follow.

\begin{proposition}\label{proposition2}
The expected throughput of the $i$-th UAV in the U2N mode~$( i \in \mathcal {M} )$ over the $k$-th subchannel is given by
\begin{equation}
ER_{i,k} = \int_{R_{th}}^\infty \mathcal{P}^{st}_{i,k} \left( r \right) dr  +  R_{th} \cdot \mathcal{P}^{st}_{i,k} \left( R_{th} \right).
\end{equation}
\end{proposition}
\begin{proof}
See Appendix~\ref{proof2}.
\end{proof}

\subsubsection{U2D Mode}
With the underlay property, a UAV in the U2D mode may suffer from the interference from co-channel U2N and U2D links. More explicitly, for the $j$-th UAV in the U2D mode~$( j \in \mathcal {N} )$ over the $k$-th subchannel, the received SINR at its corresponding mobile device can be expressed as

\begin{equation}\label{U2D_SNR}
\gamma_{j,k} = \frac{ \phi_{j,k} P^u g_j \zeta_j}{N_0 + \sum\limits_{i\in \mathcal {M}} \phi_{i,k}  P^u g_i \zeta_i + \sum\limits_{j'\in \mathcal {N}, j' \neq j} \phi_{j',k}  P^u g_{j'} \zeta_{j'}  }.
\end{equation}
Here, $\phi_{j,k}$, $\phi_{i,k}$, and $\phi_{j',k}$ denote the subchannel allocation indicators. Besides, $g_j$, $g_i$, and $g_{j'}$ are channel gains. Moreover, $\zeta_j$, $\zeta_i$, and $\zeta_{j'}$ are the small-scale fading coefficients.

Therefore, we can express the throughput of the $j$-th UAV in the U2D mode~$( j \in \mathcal {N} )$ over the $k$-th subchannel as $R_{j,k} = \log_2\left(1+\gamma_{j,k}\right)$. Likewise, given the threshold $R_{th}$, we can calculate the STP and the expected throughput for U2D transmissions by following two propositions.

\begin{proposition}\label{proposition3}
The STP of the $j$-th UAV in the U2D mode ~$( j \in \mathcal {N} )$ over the $k$-th subchannel is given by
\begin{equation}
 \mathcal{P}^{st}_{j,k}(R_{th})  =  \int_0^\infty \left\{  \mathcal{P}^{LoS}_j \cdot [1-F_{ri}(x)]+ \mathcal{P}^{NLoS}_j \cdot [1-F_{ra}(x)] \right\}  f_\chi(x) dx,
\end{equation}
where $F_{ri}(x) = 1 - Q_1(\sqrt{2K_{ri}},x \sqrt{2(K_{ri}+1)})$ and $F_{ra}(x) = 1 - e^{-x^2/2}$.
Besides, $f_\chi(x) = g(x-A_j)$, with $ g(x) = g_{\mathcal {W}[1]}(x) \ast \ldots \ast g_{\mathcal {W} [N_w]}(x)$, $A_j = \frac{N_0\kappa }{P^u g_j }$, and $\kappa = 2^{R_{th}}-1$. Here, $\mathcal {W} = \left\{ i \in \mathcal {M} | \phi_{i,k}=1  \right\} \bigcup \left\{ j' \in \mathcal {N}\setminus j | \phi_{j',k}=1  \right\}$ whose size is $N_w$, $\mathcal {W}[w]$ is the $w$-th element in $\mathcal {W}$,  $g_i(x)= ({1}/{B_j})  f_{\chi_i}({x}/{B_j})$, $B_j = \frac{ \kappa }{g_j}$, $f_{\chi_i}(y)= \mathcal{P}^{LoS}_i \cdot [({1}/{g_i}) f_{ri}(y/{g_i})] + \mathcal{P}^{NLoS}_i \cdot [({1}/{g_i}) f_{ra}(y/{g_i})]$, $f_{ri}(y)= 2\left(K_{ri}+1\right)y e^{ -\left(K_{ri}+1\right)y^2-K_{ri} } \cdot I_0\left(2\sqrt{(K_{ri}+1)K_{ri}} y\right)$, and $f_{ra}(y)=y e^{-y^2/2}$.
\end{proposition}

\begin{proposition}\label{proposition4}
The expected throughput of the $j$-th UAV in the U2D mode ~$( j \in \mathcal {N} )$ over the $k$-th subchannel is given by
\begin{equation}
ER_{j,k}  = \int_{R_{th}}^\infty \mathcal{P}^{st}_{j,k} \left( r \right) dr  +  R_{th} \cdot \mathcal{P}^{st}_{j,k} \left( R_{th} \right).
\end{equation}
\end{proposition}

The proofs for Propositions~\ref{proposition3} and~\ref{proposition4} are omitted, as the STP and the expected throughput in the U2D mode can be derived similarly to those in the U2N mode.

\subsection{AoI of UAV}
\label{Age of Information}

In this paper, we adopt the AoI to formulate the freshness of the sensory data, i.e., how timely the transmission of valid sensory data is~\cite{SRM-2012}. To be specific, in the $n$-th frame, the AoI of the $i$-th UAV~$( i \in \mathcal {M} \cup \mathcal {N} )$ is defined as
\begin{equation}\label{AoI}
 \tau_i^{(n)} = n - u_i(n),
\end{equation}
where $u_i(n)$ denote the \emph{latest} frame that the valid data transmission of the $i$-th UAV is finished. As is shown in Fig.~\ref{AoI_function}, when a UAV successfully senses its target and completely transmits the valid sensory data, the AoI will reduce to zero; otherwise, the AoI will increase with time. Equation~(\ref{AoI}) implies that the UAV with a small AoI keeps its sensory data fresher than the one with a large AoI. Therefore, to ensure the freshness of the sensory data, each UAV aims to perform sensing and transmission as quickly as possible to reduce its AoI.

\begin{figure}[!t]
\centering
\includegraphics[width=6.4in]{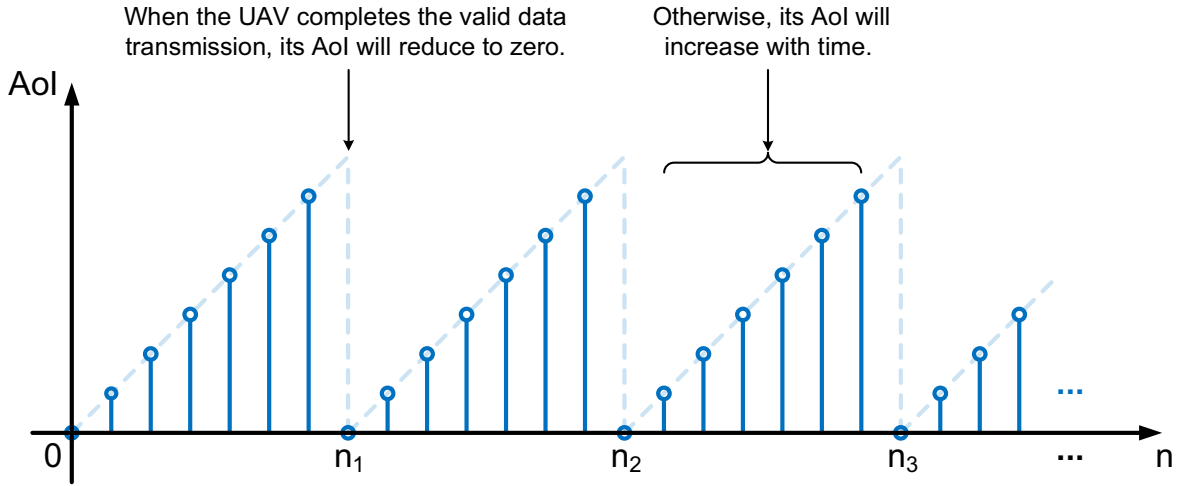}
\vspace{-12mm}
\caption{AoI as a function of the frame index $n$, in which $n_1$, $n_2$, and $n_3$ indicate the frames that the UAV completes the valid data transmission.}
\vspace{-7mm}
\label{AoI_function}
\end{figure}

\section{Joint Sensing and Transmission Protocol}
\label{Joint Sensing and Transmission Protocol}

In this section, we design a joint sensing and transmission protocol to schedule multiple UAVs performing sensing tasks. Following the overview of our proposed protocol, the subchannel allocation mechanism is elaborated on.

\subsection{Protocol Overview}
In our protocol, we assume that UAVs perform their sensing tasks in a sequence of \emph{cycles}, whose time unit is \emph{frame} with the duration $t_f$. In Fig.~\ref{protocol}, we illustrate one cycle for UAVs in the U2N mode, which also applies to those in the U2D mode. As is shown in Fig.~\ref{protocol}, each cycle consists of two stages, i.e., the sensing stage and the transmission stage. In the sensing stage, a UAV moves to a \emph{sensing location} and senses its target, while in the transmission stage, the UAV flies to a \emph{transmission location} and transmits its sensory data during the flight\footnote{When a UAV's sensing location is the same as its transmission location, the flying time in both of the stages can be zero.}. We define a \emph{stage indicator}, denoted by $\mathcal {I}_i^{(n)}$, to indicate whether the $i$-th UAV~$( i \in \mathcal {M} \cup \mathcal {N} )$ is sensing or transmitting in the $n$-th frame, whose value is $0$ or $1$ if the UAV is in the sensing stage or the transmission stage, respectively.

In each cycle, a UAV should design its trajectory by determining  its sensing and transmission locations. Since the trajectory of a UAV may be influenced by others, each UAV ought to take all UAVs' states\footnote{More explicitly, the state of a UAV contains the indexes of current frame and cycle, its current location, sensing location, transmission location, remained data size for transmission, AoI, and stage indicator. More details on the states of UAVs will be introduced in the Section~\ref{AoI Minimization Problem Formulation}.}  into consideration when it decides its sensing and transmission locations. To this end, we assume that a UAV is required to report its state to the BS before each frame, and thus the BS can have the full knowledge of all UAVs' states in the system. At the beginning of a new cycle, a UAV first sends a beacon to the BS, after which the BS broadcasts all UAVs' states to each UAV. As necessary information is provided, a UAV can then make a decision on its sensing and transmission locations in this cycle.

In what follows, we will specify the sensing stage and the transmission stage sequentially.

\begin{figure}[!t]
\centering
\includegraphics[width=6.4in]{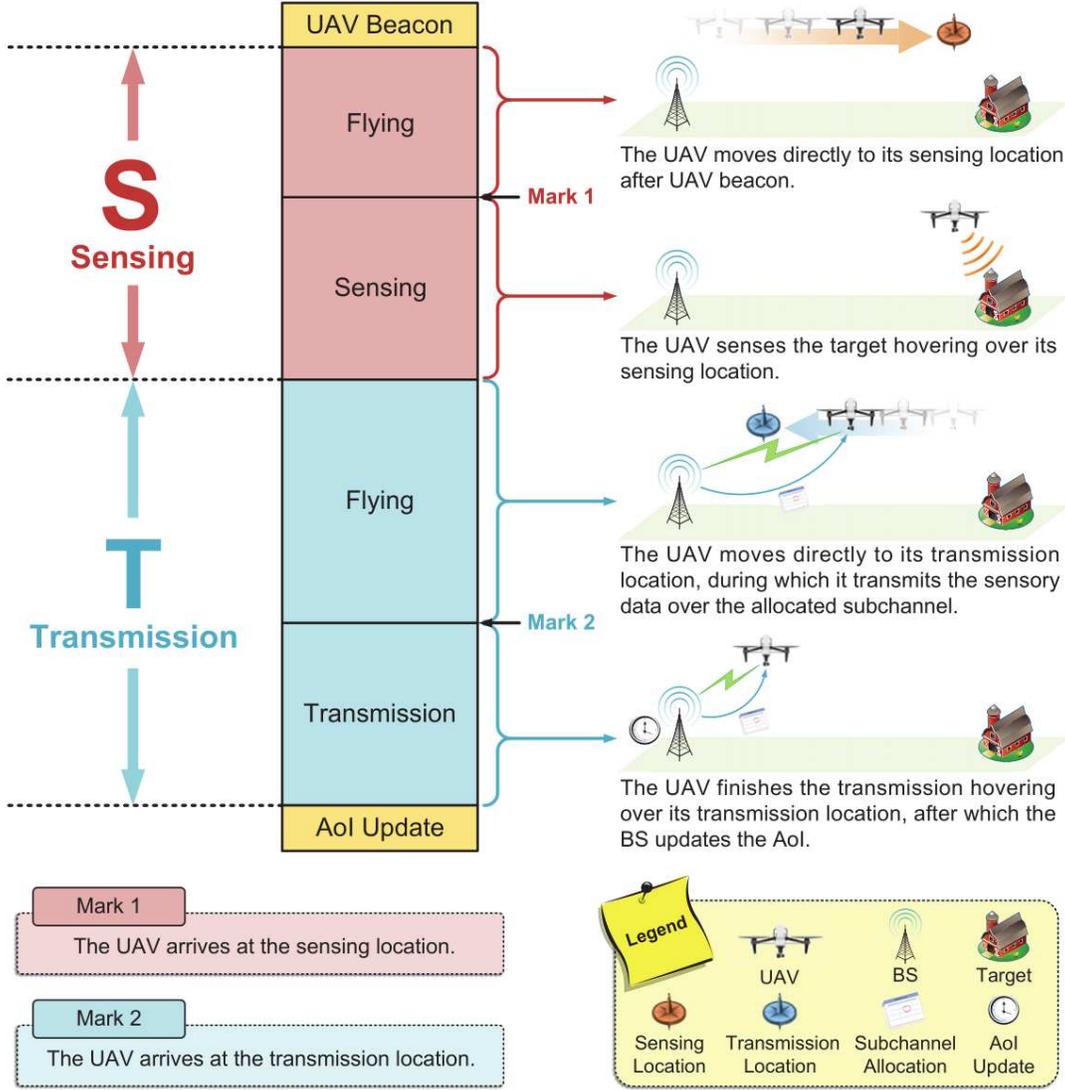}
\vspace{-9mm}
\caption{Joint sensing and transmission protocol.}
\vspace{-7mm}
\label{protocol}
\end{figure}

\subsubsection{Sensing Stage}
When a UAV has already decided its sensing location, it will fly directly towards there with the maximum flying speed $v_{max}$. Denote the location of the $i$-th UAV before the $n$-th frame as $\textbf{\emph{x}}_i ^{(n)}$, and the sensing location before this frame as $\textbf{\emph{x}}_i^{S(n)}$. We can then depict the trajectory of the UAV in the sensing stage within the $n$-th frame as
\begin{equation}\label{sensing_stage_trajectory}
\Delta\textbf{\emph{x}}_i^{S(n)} = \textbf{\emph{x}}_i ^{(n+1)} - \textbf{\emph{x}}_i ^{(n)} = \left\{ {\begin{array}{*{20}{cc}}
\frac{ \textbf{\emph{x}}_i^{S(n)} - \textbf{\emph{x}}_i ^{(n)} } { || \textbf{\emph{x}}_i^{S(n)} - \textbf{\emph{x}}_i ^{(n)} ||_2}\cdot v_{max} t_f, &\quad \textmd{if}~|| \textbf{\emph{x}}_i^{S(n)} - \textbf{\emph{x}}_i ^{(n)} ||_2 > v_{max} t_f,\\
\textbf{\emph{x}}_i^{S(n)} - \textbf{\emph{x}}_i ^{(n)},& \quad \textmd{otherwise}.
\end{array}} \right.
\end{equation}
As the UAV arrives at its sensing location, it will sense its target for one frame hovering over the sensing location. We assume that each UAV is required to collect valid sensory data with the size of at least $D^V$ in one cycle. As such, to guarantee the validness of the sensory data, the expected data size that the $i$-th UAV should collect within this cycle is given by $D^S_i = D^V/\mathcal {P}^{ss}_i$, where $\mathcal {P}^{ss}_i$ is the SSP of the $i$-th UAV. Once the sensing process is completed within the $n$-th frame, we set the remained data size before the $(n+1)$-th frame for the $i$-th UAV as $D_i^{(n+1)} = D^S_i $.

\subsubsection{Transmission Stage}
After the sensing stage, a UAV will attempt to transmit its sensory data back to the BS or the mobile device if the BS assigns a subchannel to it. However, the UAV without allocated subchannel cannot transmit its sensory data. We assume that a UAV can leave its sensing location and move directly to its transmission location for better channel condition, with the maximum flying speed $v_{max}$. Likewise, denote the current location and the transmission location of the $i$-th UAV before the $n$-th frame as $\textbf{\emph{x}}_i ^{(n)}$ and $\textbf{\emph{x}}_i^{T(n)}$, respectively. We can then describe the trajectory of the $i$-th UAV in the transmission stage within the $n$-th frame by
\begin{equation}\label{transmission_stage_trajectory}
\Delta\textbf{\emph{x}}_i^{T(n)} = \textbf{\emph{x}}_i ^{(n+1)} - \textbf{\emph{x}}_i ^{(n)} = \left\{ {\begin{array}{*{20}{cc}}
\frac{ \textbf{\emph{x}}_i^{T(n)} - \textbf{\emph{x}}_i ^{(n)} } { || \textbf{\emph{x}}_i^{T(n)} - \textbf{\emph{x}}_i ^{(n)} ||_2}\cdot v_{max} t_f, &\quad \textmd{if}~|| \textbf{\emph{x}}_i^{T(n)} - \textbf{\emph{x}}_i ^{(n)} ||_2 > v_{max} t_f,\\
\textbf{\emph{x}}_i^{T(n)} - \textbf{\emph{x}}_i ^{(n)},& \quad \textmd{otherwise}.
\end{array}} \right.
\end{equation}
To further improve the efficiency, we also assume that a UAV can transmit its sensory data during the flight. When the UAV arrives at its transmission location, it will continue transmitting its sensory data hovering over there until all of the sensory data is transmitted. The expected data size that the $i$-th UAV over the $k$-th subchannel can transmit within the $n$-th frame is given by $D^{T(n)}_{i,k} = t_f \cdot ER_{i,k}^{(n)}$, where $ER_{i,k}^{(n)}$ is the expected throughput in this frame.  As such, we can calculate the remained data size for the $i$-th UAV before the $(n+1)$-th frame as $D_i^{(n+1)} = \max\{ 0, D_i^{(n)}-D^{T(n)}_{i,k}\}$.

Finally, as a UAV completely transmits all of its sensory data and finishes the transmission stage, the BS or the mobile device evaluates the validness of the received data, and then updates the AoI for the UAV, which ends the current cycle.

\subsection{Subchannel Allocation Mechanism}

At each frame, the BS performs subchannel allocation for all UAVs in the transmission stage. For the sake of timeliness, we assume that the BS schedules subchannels to maximize the sum of expected data size transmitted in the system within this frame. To be specific, in the $n$-th frame, the sum of expected data size is given by
\begin{equation}
D^{(n)} = \sum^{M+N}_{i=1}\sum^{K}_{k=1} D^{T(n)}_{i,k}
\end{equation}
Therefore, the subchannel allocation matrix in the $n$-th frame can be obtained by $\Phi^{(n)} = \mathop {\max \arg} \limits_{ \Phi } D^{(n)} $.

An available subchannel allocation matrix should satisfy following constraints:
\begin{itemize}
  \item Based on our proposed protocol, the BS does not allocate subchannels to UAVs in the sensing stage, i.e., for all $i$ satisfying $\mathcal {I}_i^{(n)} = 0$, we have $\phi_{i,k} = 0$, $\forall k \in \mathcal {K}$;
  \item According to Section~\ref{System Model}, each subchannel can be utilized for at most one U2N link and $V_s$ U2D links, and thus we have $\sum^M_{i=1} \phi_{i,k} \leq 1$ and $\sum^N_{i=1} \phi_{i,k} \leq V_s$, $\forall k \in \mathcal {K}$. In addition, each UAV can only access one subchannel  at a time, i.e., $\sum^K_{k=1} \phi_{i,k} \leq 1$, $\forall i \in \mathcal {M} \cup \mathcal {N}$.
\end{itemize}

Finding the optimal subchannel allocation matrix under above constraints can be regarded as a classic mixed-integer non-linear programming~(MINLP) problem~\cite{D-1976}, which can be reformulated as a many-to-one two-sided matching problem~\cite{ECABA-2011,HTLZ-2014}, and then be efficiently solved by utilizing the swap matching algorithm in~\cite{HYL-2017}. The stability and the optimality of the obtained subchannel allocation matrix can be guaranteed. It is also worthwhile to mention that the subchannel allocation matrix in a frame is only determined by the locations of all UAVs in this frame. Since the subchannels are limited, each UAV has the incentive to compete with other UAVs by designing its trajectory.

\section{Age of Information Minimization Problem Formulation}
\label{AoI Minimization Problem Formulation}

In this section, we investigate on the AoI minimization problem for UAVs in the cellular Internet of UAVs, which can be regarded as a MDP. Since the interference among UAVs is not observable, traditional model-based methods are infeasible to tackle with this MDP. Therefore, we reformulate this problem using multi-agent RL to make it solvable.

\subsection{AoI Minimization Problem}

In this paper, we aim to minimize the average AoI in the system by designing the trajectories of all UAVs in the future $N_f$ frames. Since a UAV's trajectory can be determined by its sensing and transmission locations as it performs cycles, the AoI minimization problem can be formulated as optimizing UAVs' sensing and transmission locations at each frame, i.e.,
\begin{equation}\label{AoI problem}
\begin{split}
\textbf{P1:}~~\mathop {\min } \limits_{ \textbf{\emph{x}}_i^{S(n)}, \textbf{\emph{x}}_i^{T(n)}  } ~~& \frac{1}{N_f}\sum^{N_f}_{ n = 1 }\tau_i ^{(n)} \\
s.t. ~~& \textbf{\emph{x}}_i^{S(n)}, \textbf{\emph{x}}_i^{T(n)} \in \mathcal {\mathcal {X}}, \\
&   \mathcal {P}^{ss}_i \left( \textbf{\emph{x}}_i^{S(n)}  \right) > 0.
\end{split}
\end{equation}
Here, $\mathcal {\mathcal {X}}$ denotes the set of all possible locations in the space, and $n \in [1,N_f]$.

As the AoI of each UAV in the future frames is influenced by the trajectories of all UAVs, it is challenging to solve problem~(\ref{AoI problem}). Fortunately, in this problem, UAVs' states at a frame are only determined by their states and decisions in the last frame, and thus this problem can be regarded as an MDP~\cite{RA-1998}. Therefore, we can adopt multi-agent RL to solve this problem efficiently. In what follows, we will first analyze problem~(\ref{AoI problem}) under the MDP framework. After that, we will reformulate this problem by adopting multi-agent RL.

\subsection{Multi-agent RL Formulation}

In our system, each UAV is regarded as an \emph{agent}, and all of the network setting~(including the BS, sensing target, and mobile devices) is regarded as the \emph{environment}. According to~\cite{KMMA-2017}, we can characterize all UAVs by a tuple\footnote{The standard form of the tuple also includes a \emph{discount factor} which evaluates the timeliness of the reward in the future. In our problem, we assume that the discount factor equals to one. Therefore, we omit it in the tuple for simplicity.} $ < \mathcal {S}, \{\mathcal {A}_i\}_{i \in\mathcal {M} \cup \mathcal {N}}, \mathcal {P}, \mathcal \{\mathcal {R}_i\}_{i \in\mathcal {M} \cup \mathcal {N}}>$, in which
\begin{itemize}
  \item $\mathcal {S}$ is the \emph{state space} including all possible states of UAVs in the system at each frame;
  \item $\mathcal {A}_i ~( i \in\mathcal {M} \cup \mathcal {N} )$ is the \emph{action space} of the $i$-th UAV, which consists of all available actions of the UAV at the each frame;
  \item $\mathcal {P}$ is the \emph{state transition function}, which maps the state spaces and the action spaces of all UAVs in the current frame to their state spaces in the next frame;
  \item $\mathcal {R}_i~ ( i \in\mathcal {M} \cup \mathcal {N} )$ is the \emph{reward function} of the $i$-th UAV, which maps the state spaces and the action spaces of the UAV in the current frame to its expected reward;
\end{itemize}

In the following, we will elaborate on the above elements sequentially.

\subsubsection{State Space}
We define the state of UAVs in the system before the $n$-th frame as $\textbf{\emph{s}}^{(n)} = \{ \textbf{\emph{s}}_i^{(n)} \}_{i \in\mathcal {M} \cup \mathcal {N}}$, where $\textbf{\emph{s}}_i ^{(n)}= \left( n, c_i^{(n)}, \textbf{\emph{x}}_i^{(n)} , \textbf{\emph{x}}_i^{S(n)} , \textbf{\emph{x}}_i^{T(n)}, D_i ^{(n)},  \tau_i^{(n)}, \mathcal {I}_i^{(n)} \right) $ indicating the state of the $i$-th UAV. Here, $n$ is the frame index, $c_i^{(n)}$ is the cycle index, $\textbf{\emph{x}}_i^{(n)}$ is the current location, $\textbf{\emph{x}}_i^{S(n)}$ is the sensing location, $\textbf{\emph{x}}_i^{T(n)}$ is the transmission location, $D_i^{(n)}$ is the remained data size, $\tau_i^{(n)}$ is the AoI, and $\mathcal {I}_i^{(n)}$ is the stage indicator.

\subsubsection{Action Space}
We define the actions of a UAV as its decisions on the sensing and the transmission locations. To be specific, the action of the $i$-th UAV within the $n$-th frame is expressed as $\textbf{\emph{a}}_i^{(n)} = \left( \textbf{\emph{a}}_i^{S(n)}, \textbf{\emph{a}}_i ^{T(n)}\right)$, in which $\textbf{\emph{a}}_i^{S(n)}$ and $\textbf{\emph{a}}_i^{T(n)}$ denote the sensing and the transmission locations in this frame.

According to our proposed protocol, a UAV makes new decisions on its sensing and the transmission locations only at the beginning of cycles. Therefore, we can conclude that when the $i$-th UAV finishes the data transmission in the~$n$-th frame, it will make a new decision in the $(n+1)$-th frame; otherwise, it will keep its decision unchanged in the $(n+1)$-th frame. As such, we can express the available action set for the $i$-th UAV within the $(n+1)$-th frame as
\begin{equation}
\mathcal {A}_i\left(\emph{\textbf{s}}^{(n+1)}\right) = \left\{ {\begin{array}{*{20}{cc}}
\left\{  (\textbf{\emph{a}}^S_i,\textbf{\emph{a}}^T_i) \in \mathcal {X}^2 | \mathcal {P}^{ss}_i (\textbf{\emph{a}}^S_i ) > 0 \right\}, & \quad \textmd{if}~D_i^{(n+1)} = 0~\textmd{and}~\mathcal {I}_i^{(n)} = 1 ,\\
\textbf{\emph{a}}_i^{(n)}, & \quad \textmd{otherwise}.
\end{array}} \right.
\end{equation}

\subsubsection{State Transition Function}
We define the state transition from the $n$-th frame to the~$(n+1)$-th frame for the $i$-th UAV as follow.

\begin{itemize}
  \item The cycle index before the~$(n+1)$-th frame increases by one only when the UAV finishes the transmission of sensory data within the~$n$-th frame, i.e.,
\end{itemize}
\begin{equation}
c_i^{(n+1)} = \left\{ {\begin{array}{*{20}{cc}}
c_i^{(n)}+1, & \quad \textmd{if}~D_i^{(n+1)} = 0~\textmd{and}~\mathcal {I}_i^{(n)} = 1 ,\\
c_i^{(n)} , & \quad \textmd{otherwise}.
\end{array}} \right.
\end{equation}

\begin{itemize}
  \item The UAV's current location before the~$(n+1)$-th frame is determined by the UAV's location and stage before the~$n$-th frame. Based on the current action $\textbf{\emph{a}}_i^{(n)}$, we can obtain the movement of the UAV within the sensing and the transmission stages, i.e., $\Delta\emph{\textbf{x}}^{S(n)}_i$ and $\Delta\emph{\textbf{x}}_i^{T(n)}$, from~(\ref{sensing_stage_trajectory}) and~(\ref{transmission_stage_trajectory}), respectively. Then, we have
\end{itemize}
\begin{equation}
\emph{\textbf{x}}_i^{(n+1)} = \left\{ {\begin{array}{*{20}{cc}}
\emph{\textbf{x}}_i^{(n)} + \Delta\emph{\textbf{x}}^{S(n)}_i, & \quad \textmd{if}~\mathcal {I}_i^{(n)} = 0,\\
\emph{\textbf{x}}_i^{(n)} + \Delta\emph{\textbf{x}}^{T(n)}_i, & \quad \textmd{otherwise}.
\end{array}} \right.
\end{equation}

\begin{itemize}
  \item The UAV's sensing and transmission locations before the $(n+1)$-th frame can be obtained from the UAV's action in the $n$-th frame, i.e.,
\end{itemize}
\begin{equation}
\emph{\textbf{x}}^{S(n+1)}_i = \emph{\textbf{a}}^{S(n)}_i, \quad \emph{\textbf{x}}^{T(n+1)}_i = \emph{\textbf{a}}^{T(n)}_i
\end{equation}

\begin{itemize}
  \item The UAV's remained data size before the $(n+1)$-th frame is determined by the UAV's stage and location before the $n$-th frame, i.e.,
\end{itemize}
\begin{equation}
D_i^{(n+1)} = \left\{ {\begin{array}{*{20}{cc}}
\max\{ 0, D_i^{(n)}-D_{i,k}^{T(n)} \}, & \quad \textmd{if}~\mathcal {I}_i^{(n)} = 1,\\
D^S_i, & \quad \textmd{if}~\emph{\textbf{x}}_i^{(n)} = \emph{\textbf{x}}^{S(n)}_i~\textmd{and}~\mathcal {I}_i^{(n)} = 0,\\
0, & \quad \textmd{otherwise}.
\end{array}} \right.
\end{equation}

\begin{itemize}
  \item The UAV's AoI is updated and reduced to zero in the $(n+1)$-th frame only when the UAV finishes the data transmission within the $n$-th frame; otherwise, the AoI increases with time. Thus, we have
\end{itemize}
\begin{equation}
\tau_i^{(n+1)} = \left\{ {\begin{array}{*{20}{cc}}
0, & \quad \textmd{if}~D_i^{(n+1)} = 0~\textmd{and}~\mathcal {I}_i^{(n)} = 1,\\
n +1 - u_i(n+1) , & \quad \textmd{otherwise}.
\end{array}} \right.
\end{equation}

\begin{itemize}
  \item The UAV's stage indictor will switch to the sensing stage before the $(n+1)$-th frame when it finishes the data transmission within the $n$-th frame. Besides, the stage indictor will switch to the transmission stage before the $(n+1)$-th frame if the UAV collects the sensory data within the $n$-th frame. Otherwise, the stage indictor will remain unchanged. Therefore, the stage indictor before the $(n+1)$-th frame can be given by
\end{itemize}
\begin{equation}
\mathcal {I}_i^{(n+1)} = \left\{ {\begin{array}{*{20}{cc}}
0, & \quad \textmd{if}~D_i^{(n+1)} = 0~\textmd{and}~\mathcal {I}_i^{(n)} = 1,\\
1, & \quad \textmd{if}~D_i^{(n+1)} = D^S_i~\textmd{and}~\mathcal {I}_i^{(n)} = 0,\\
\mathcal {I}_i^{(n)}, & \quad \textmd{otherwise}.
\end{array}} \right.
\end{equation}

\subsubsection{Reward Function}

We define the reward of the $i$-th UAV in the $n$-th frame as the minus AoI within this frame, i.e.,
\begin{equation}
r_i\left(\emph{\textbf{s}}^{(n)},\emph{\textbf{a}}_i^{(n)}\right) = - \tau_i^{(n)},
\end{equation}
Therefore, each UAV are motivated to minimize its AoI by making decisions on its sensing and transmission locations.

In our system, the \emph{policy} of a UAV is defined as a mapping from its state space to its action space, denoted by $\pi$. To be specific, the policy of the $i$-th UAV can be expressed by $\emph{\textbf{a}}_i= \pi_i(\emph{\textbf{s}})$, where $\emph{\textbf{s}}$ is the state of all UAVs in the system, and $\emph{\textbf{a}}_i$ is the action of the $i$-th UAV.
Before each frame, the $i$-th UAV first observes\footnote{Actually, a UAV cannot know the states of other UAVs by itself. However, before making decision in a cycle, the UAV can observe the states of other UAVs from the BS by sending a beacon. After making the decision, the UAV will keep its decision unchanged till the end of the current cycle. } the current state of all UAVs $\emph{\textbf{s}}$, and then takes an action $\emph{\textbf{a}}_i$ according to its policy $\pi_i$. After that, the UAV receives a reward $r_i$ and then observes the next state $\emph{\textbf{s}}'$, namely the state of all UAVs before the next frame. Therefore, the AoI minimization problem in~(\ref{AoI problem}) can be reformulated as maximizing the accumulated rewards of all UAVs in the system by optimizing their policies, i.e.,
\begin{equation}\label{RL problem}
\begin{split}
\textbf{P2:}~~\mathop {\max } \limits_{ \pi_i } ~~& \frac{1}{N_f}\sum^{N_f}_{ n = 1 } r_i \left(\emph{\textbf{s}}^{(n)},\emph{\textbf{a}}_i^{(n)} \right) \\
s.t. ~~& \pi_i \left( \textbf{\emph{s}}^{(n)} \right) \in \mathcal {A}_i\left(\emph{\textbf{s}}^{(n)}\right), \quad n \in [1,N_f]. \\
\end{split}
\end{equation}

\section{Algorithm Design by Multi-agent Deep Reinforcement Learning}
\label{Algorithm Design by Multi-agent DRL}
In our system, UAVs are incapable of obtaining enough information to specify their state transition functions. As such, a model-free RL algorithm, which does not require the prior information on state transition functions, is needed to solve this problem. Since UAVs' sensing and transmission locations are continuous-valued, the state-action space in problem~(P2) is infinite, which makes value-based algorithms infeasible\footnote{Specifically, value-based algorithms, e.g., the Q-learning algorithm~\cite{CP-1992}, can only be applied for the problems where agents' states and actions are discrete-valued. }. Therefore, we adopt a policy-based algorithm\footnote{It is worth mentioning that the policy gradient algorithm~\cite{RDSY-1999} is not adopted, since it may suffer high variance when the policies of multiple agents are optimized simultaneously.}, namely the Actor-Critic~(AC) algorithm~\cite{VJ-2000}, to cope with problem~(P2). In the AC algorithm, there exist the \emph{actor networks} for action selection and the \emph{critic networks} for action evaluation. To further accelerate the convergence, deep Q networks~(DQNs)~\cite{nature-2015,HG-2018} are utilized for value function approximation in actor and critic networks, which is also known as the DDPG algorithm~\cite{TJANTYDD-2016}. Once the actor and the critic networks of a UAV are well-trained, the policy of this UAV can be obtained.

In this section, we first propose a DDPG-based multi-UAV trajectory design algorithm to optimize multiple UAVs' policies. After that, we introduce the training process of the proposed algorithm. For simplicity, the frame index $n$ is omitted in the following notations.

\subsection{Algorithm Design}

We define the Q-value of the $i$-th UAV, denoted by $Q_i \left( \emph{\textbf{s}},\emph{\textbf{a}}_i \right)$, as the accumulated reward when it takes action $\emph{\textbf{a}}_i$ at state $\emph{\textbf{s}}$ and follows its policy $\pi_i$ afterwards. Specifically, we have
\begin{equation}
Q_i(\emph{\textbf{s}},\emph{\textbf{a}}_i) = r_i(\emph{\textbf{s}}, \emph{\textbf{a}}_i) + \mathbb{E}_{ \emph{\textbf{s}}' \sim \mathcal {P}\left(\emph{\textbf{s}},\emph{\textbf{a}}_i,\pi_{-i}(\emph{\textbf{s}})\right), \emph{\textbf{a}}_i' \sim \pi_i(\emph{\textbf{s}}') }
\left[ Q_i(\emph{\textbf{s}}', \emph{\textbf{a}}_i')\right],
\end{equation}
in which $\mathcal {P}\left(\emph{\textbf{s}},\emph{\textbf{a}}_i,\pi_{-i}(\emph{\textbf{s}})\right)$ denotes the state transition function, with $\pi_{-i}(\emph{\textbf{s}})$ implies the policies of the UAVs except the $i$-th UAV. At each state, the optimal policy of each $i$-th UAV is to select the action which can maximize its Q-value\cite{RA-1998}. Thus, we can describe the optimal policy of the $i$-th UAV at state $\emph{\textbf{s}}$ as
\begin{equation}\label{optimal_policy}
\pi^*_i(\emph{\textbf{s}}) = \mathop {\arg \max }\limits_{\emph{\textbf{a}}_i \in \mathcal {A}_i(\emph{\textbf{s}})} Q_i(\emph{\textbf{s}},\emph{\textbf{a}}_i).
\end{equation}
Therefore, to obtain the optimal policy $\pi^*_i$, we have to specify the Q-function $Q_i(\emph{\textbf{s}},\emph{\textbf{a}}_i)$.

Due to infinite state-action space, we can hardly obtain exact Q-functions for each UAV. Instead, we adopt two deep neural networks~(DNNs), including an actor network and a critic network, to approximate Q-functions for each UAV. More explicitly, the actor network of the $i$-th UAV, denoted by $\mu_i$,  determines its action $\emph{\textbf{a}}_i $ based on the current state $\emph{\textbf{s}}$, i.e., $ \emph{\textbf{a}}_i = \mu_i(\emph{\textbf{s}};\Theta^{\mu}_i)$, in which $\Theta^{\mu}_i$ is the weight of the actor network. In addition, the critic network of the $i$-th UAV, denoted by $Q_i$, approximates its Q-value given the current state $\emph{\textbf{s}}$ and the determined action $\emph{\textbf{a}}_i$, i.e., $ Q_i(\emph{\textbf{s}},\emph{\textbf{a}}_i;\Theta^Q_i)$, with $\Theta^Q_i$ being the weight of the critic network. By this means, when these two networks are well-trained, the UAV's policy at any state is given by the output of the actor network, whose Q-value is evaluated by the output of the critic network.

\subsection{Algorithm Training}

In order to train the actor and the critic networks, UAVs have to record their experience as training samples. As UAVs perform their sensing tasks in a sequence of cycles, each sample should contain the experience of UAVs in a whole cycle. Therefore, different from traditional DDPG algorithm in which the training sample is generated after each state transition, in our system, the training sample is generated only after each cycle. When the $i$-th UAV finishes a cycle, it will generate a sample specified by a tuple $<\tilde{\emph{\textbf{s}}},\tilde{\emph{\textbf{a}}}_i,\tilde{r}_i,\tilde{\emph{\textbf{s}}}'>$, which contains the initial state of the current cycle $\tilde{\textbf{\emph{s}}}$, the action taken in the current cycle $\tilde{\textbf{\emph{a}}}_i$, the accumulated reward within the current cycle $\tilde{r}_i$, and the initial state of the next cycle $\tilde{\textbf{\emph{s}}}'$.

Besides, to suppress the temporal correlation among  training samples, we utilize the \emph{experience replay}~\cite{HG-2018} to generate training sets. More explicitly, we store all of the training samples of the $i$-th UAV in a replay memory, denoted by $\mathcal {RM}_i$, whose maximum size is denoted by $N_{rm}$. When the number of training samples in a replay memory exceeds the maximum size, the fresh samples will replace the out-of-date ones. Before each time of training, we randomly choose a mini-batch with $N_{mini}$ samples from replay memory $\mathcal {RM}_i$ as the training set for the $i$-th UAV, denoted by $\mathcal {D}_i$.

In the training of the actor and the critic networks for the $i$-th UAV, i.e., $\mu_i$ and $Q_i$, we adopt the two separate networks, also known as the \emph{target networks}~\cite{nature-2015}, to generate the training targets. Specifically, the target actor network and the target critic network for the $i$-th UAV are denoted by $\widehat{\mu}_i$ and $\widehat{Q}_i$, whose weights are $\Theta^{\widehat{\mu}}_i$ and $\Theta^{\widehat{Q}}_i$, respectively. To train actor network $\mu_i$, we take steps in the gradient direction of a performance evaluation function $J_i(\Theta^{\mu}_i)$, i.e.,
\begin{equation}\label{actor update}
\Theta^{\mu}_i \leftarrow \Theta^{\mu}_i + \alpha \nabla_{\Theta^{\mu}_i } J_i(\Theta^{\mu}_i).
\end{equation}
Based on~\cite{TJANTYDD-2016}, the gradient $\nabla_{\Theta^{\mu}_i } J_i(\Theta^{\mu}_i)$ can be calculated by
\begin{equation}
\nabla_{\Theta^{\mu}_i } J_i(\Theta^{\mu}_i) = \mathbb{E}_{(\tilde{\textbf{\emph{s}}},\tilde{\textbf{\emph{a}}}_i) \sim \mathcal {D}_i}\left[ \nabla_{ \emph{\textbf{a}}_i } Q_i(\emph{\textbf{s}},\emph{\textbf{a}}_i;\Theta^Q_i )|_{\emph{\textbf{s}} = \tilde{\emph{\textbf{s}}}_i, \emph{\textbf{a}}_i = \tilde{\emph{\textbf{a}}}_i }  \cdot \nabla_{ \Theta^{\mu}_i  } \mu_i(\emph{\textbf{s}};\Theta^{\mu}_i ) |_{\emph{\textbf{s}} = \tilde{\emph{\textbf{s}}}_i}  \right].
\end{equation}
In addition, we train critic network $Q_i$ by taking steps in the gradient direction of a loss function $L_i(\Theta^Q_i)$, i.e.,
\begin{equation}\label{critic update}
\Theta^Q_i \leftarrow \Theta^Q_i + \alpha \nabla_{\Theta^Q_i } L_i(\Theta^Q_i).
\end{equation}
Here, the loss function $L_i(\Theta^Q_i)$ is expressed as
\begin{equation}
L_i(\Theta^Q_i) = \mathbb{E}_{(\tilde{\textbf{\emph{s}}},\tilde{\textbf{\emph{a}}}_i,\tilde{r}_i,\tilde{\emph{\textbf{s}}}') \sim \mathcal {D}_i} \left[  ( y_i - Q_i(\tilde{\emph{\textbf{s}}},\tilde{\emph{\textbf{a}}}_i;\Theta^Q_i ) )^2  \right],
\end{equation}
in which the target $y_i$ is given by
\begin{equation}
y_i = \tilde{r}_i + \widehat{Q}_i(\tilde{\emph{\textbf{s}}}',\widehat{\mu}_i(\tilde{\emph{\textbf{s}}}';\Theta^{\widehat{\mu}}_i);\Theta^{\widehat{Q}}_i  ).
\end{equation}
It is worth mentioning that target networks $\widehat{\mu}_i$ and $\widehat{Q}_i$ are not trained through the above methods. Instead, we update the weights of target networks by following rules:
\begin{equation}\label{target_actor_update}
\Theta^{\widehat{\mu}}_i \leftarrow \nu \Theta^{\mu}_i + (1-\nu) \Theta^{\widehat{\mu}}_i,
\end{equation}
\begin{equation}\label{target_critic_update}
\Theta^{\widehat{Q}}_i  \leftarrow \nu \Theta^{Q}_i + (1-\nu) \Theta^{\widehat{Q}}_i,
\end{equation}
in which $\nu$ is the update rate.

We present our proposed DDPG-based multi-UAV trajectory design algorithm in Algorithm~\ref{Trajectory Design Algorithm}. For the sake of \emph{exploration}~\cite{TJANTYDD-2016}, we construct the exploration policy of the $i$-th UAV, denoted by $\mu'_i$, by adding noise sampled from a stochastic process $\mathcal {E}$ to actor policy $\mu_i$. Specifically, at state $\emph{\textbf{s}}$, the exploration policy of the $i$-th UAV is expressed by
\begin{equation}\label{action_selection}
\mu'_i(\emph{\textbf{s}}) = \mu_i(\emph{\textbf{s}};\Theta^{\mu}_i) + \varepsilon.
\end{equation}
Here, $\varepsilon$ is the noise sampled from stochastic process $\mathcal {E}$, which is generated by the Ornstein-Uhlenbeck process~\cite{GL-1930}.

\begin{algorithm}[!t]
  \caption{DDPG-based multi-UAV trajectory design algorithm of the $i$-th UAV.}
  \label{Trajectory Design Algorithm}
  \begin{algorithmic}[1]
    \REQUIRE
	    The structures of actor network, critic network, and their target networks; Number of episodes $N_{epi}$; Number of frames per episode $N_f$; Initial state $\textbf{\emph{s}}$.
    \ENSURE
        Policy $\pi_i$;
    \STATE Initialize all networks;
    \STATE Initialize replay memory $\mathcal {RM}_i$;
    \FOR{$n_{epi}= 1 : N_{epi}$ }
    \FOR{$n = 1 : N_f$ }
    \STATE Observe the current state $\textbf{\emph{s}}$;
    \IF{$\tau_i = 0$}
    \STATE Select action according to~(\ref{action_selection});
    \ELSE
    \STATE Keep on the current action;
    \ENDIF
    \STATE Observe reward $r_i$, and transit to the next state $\emph{\textbf{s}}'$;
    \STATE Accumulate reward $\tilde{r}_i = \tilde{r}_i + r_i$;
    \IF{$\tau_i' = 0$}
    \STATE Store the experience into replay memory $\mathcal {RM}_i$;
    \STATE Set reward $\tilde{r}_i = 0$;
    \STATE Sample a mini-batch $\mathcal {D}_i$ from replay memory $\mathcal {RM}_i$;
    \STATE Train actor and critic networks according to~(\ref{actor update}) and~(\ref{critic update});
    \STATE Update target networks according to~(\ref{target_actor_update}) and~(\ref{target_critic_update}).
    \ENDIF
    \ENDFOR
	\ENDFOR
	\end{algorithmic}
\end{algorithm}

\section{System Performance Analysis}
\label{System Performance Analysis}

In this section, we first analyze the convergency and the complexity of our proposed multi-UAV trajectory design algorithm, and then remark some properties on the AoIs of UAVs.

\subsection{Algorithm Analysis}
\subsubsection{Convergency}
In Algorithm~\ref{Trajectory Design Algorithm}, we adopt the gradient descend method to train actor network $\mu_i$ and critic network $Q_i$, in which the learning rate is exponentially decayed with iterations. Therefore,  the weights $\Theta^{\mu}_i$ and $\Theta^Q_i$ will converge after a finite number of iterations, which guarantees the convergency of our algorithm. Actually, as referred in~\cite{UWC-2019}, the convergency of a neural network can hardly be theoretically analyzed before training. The reason lies in that the convergence of a neural network is highly dependent on the hyperparameters during the training process, in which the quantitative relationship between the network convergency and the hyperparameters is sophisticated. Instead, in our paper, we  show the convergency of our algorithm through simulation.

\subsubsection{Complexity}
The time complexity for training actor network $\mu_i$ and critic network $Q_i$ is determined by the number of operations in each iteration during the update. Assume that a network has $N_h$ hidden layers, whose numbers of neurons are denoted by $q_i$, $i=1,...,N_h$. As such, the time complexity in each iteration can be given by $\mathcal {O}\left(\sum_{i=1}^{N_h-1} q_i q_{i+1}\right)$. Here, we ignore the operations at input and output layers, as their numbers are trivial compared with that at the hidden layers. When all hidden layers in the network have the same amount of neurons, denoted by $q$, the time complexity can be reduced to $\mathcal {O}\left((N_h-1) q^2\right) = \mathcal {O}\left(q^2\right)$.

\subsection{AoI Analysis}
In the system, the AoIs of UAVs are jointly determined by two main factors, namely the valid data size requirement $D^V$ and the number of subchannels $K$.  In what follows, we will introduce two propositions to analyze the effects of them on the average AoI in the system.

\begin{proposition}\label{Remark1}
If the number of subchannel $K$ is given, the increase of valid data size requirement $D^V$ approximately leads to a linear increase of the average AoI.

\end{proposition}
\begin{proof}
According to the protocol in Section~\ref{Joint Sensing and Transmission Protocol}, the expected throughput of a UAV is determined by the current locations of all UAVs, which is hardly influenced by $D^V$. When $D^V$ increases, the expected throughput of UAVs can be approximately regarded as constant, and thus the average duration of performing a cycle, denoted by $N_c$, will linearly increase, i.e., $N_c \propto D^V$. Moreover, based on the definition of AoI in Section~\ref{System Model}, the sum AoI of each UAV in a cycle can be expressed as a quadratic function of $N_c$, i.e., $\sum^{N_c}_{n=1} n = \frac{N_c^2+N_c}{2}$. Then, we can approximate the average AoI given in problem~(P1) by
\begin{equation}
\frac{1}{N_f}\sum^{N_f}_{ n = 1 }\tau_i ^{(n)}  \approx \frac{1}{N_f} \cdot \frac{N_f}{N_c}  \cdot  \frac{N_c^2+N_c}{2} = \frac{N_c+1}{2} \propto D^V.
\end{equation}
Therefore, we can conclude that the average AoI linearly increases with $D^V$.
\end{proof}

\begin{proposition}\label{Remark2}
If the valid data size requirement $D^V$ is given, the average AoI decreases with the number of subchannels $K$, and then becomes saturated. Specifically, when $K$ is below $\max\{M,{N}/{V_s}\}$, the average AoI sharply decreases with $K$. As $K$ increases from $\max\{M,{N}/{V_s}\}$ to $(M+N)$, the average AoI slightly decreases with $K$. When $K$ is larger than $(M+N)$, the average AoI remains unchanged.
\end{proposition}

\begin{proof}
According to the subchannel allocation mechanism in Section~\ref{Joint Sensing and Transmission Protocol}, a subchannel can only be utilized for one U2N link and $V_s$ U2D links at a time. When $K$ is below $\max\{M,{N}/{V_s}\}$, there exist UAVs cannot be allocated to subchannels, and thus they cannot transmit the sensory data timely, leading to an extremely high average AoI. In this case, a higher $K$ may lead to fewer UAVs without the allocated subchannels, which leading to a much lower average AoI. In addition, when $K$ exceeds $\max\{M,{N}/{V_s}\}$, all UAVs can be allocated to subchannels, and thus the average AoI changes much slighter with $K$ than the case where $K<\max\{M,{N}/{V_s}\}$. As the increase of $K$, the mutual interference among UAVs decreases, since more subchannels can be utilized. Consequently, the expected throughput of UAVs increases, resulting in the decrease of average AoI. Moreover, when $K$ reaches to $(M+N)$, each U2N or U2D link can exclusively occupy one subchannel, and thus the mutual interference among UAVs can be avoided. After that, the increase of $K$ does not improve the expected throughput of UAVs. Therefore, the average AoI remains unchanged when $K$ exceeds $(M+N)$.
\end{proof}

\section{Simulation Result}
\label{Simulation Result}

\begin{table}[!t]
\centering
\caption{Parameters for Simulation} \label{parameters}
\vspace{-1mm}
\begin{tabular}{|p{2.2in}|p{0.6in}|p{2.0in}|p{0.8in}|}
 \hline
 \textbf{Parameters} & \textbf{Values} & \textbf{Parameters} & \textbf{Values} \\
 \hline
 \hline
  Number of UAVs in the U2N mode $M$ & 5 & Maximum sensing angle $\varphi$ & 30$^\circ$ \\
  \hline
  Number of UAVs in the U2D mode $N$ & 5 & Sensing factor $\lambda$ & 0.005 (s$\cdot$m)$^{-1}$\\
  \hline
  Number of subchannels $K$ & 5 & QoS requirement $R_{th}$ & 1 bps/Hz\\
  \hline
  Transmit power of UAVs $P^u$ & 10 dBm &  Valid data size requirement $D^V$ & 10 bit/Hz\\
  \hline
  Noise power $N_0$ & -85 dBm &  Maximum U2D links per subchannel $V_s$ & 3\\
  \hline
  Carrier frequency $f_c$ & 2 GHz &  Number of episodes $N_{epi}$ & 500\\
  \hline
  Height of the BS $H_0$ & 10 m & Number of frames per episode $N_f$ & 300\\
  \hline
  Minimum flying altitude of UAVs $h_{min}$ & 50 m & Size of mini-batch $N_{mini}$ & 64\\
  \hline
  Maximum flying altitude of UAVs $h_{max}$ & 150 m & Size of replay memory $N_{rm}$ & 10000\\
  \hline
  Maximum flying speed of UAVs $v_{max}$ & 15 m/s &  Learning rate $\alpha$ & 0.001\\
  \hline
  Duration of a frame $t_f$ & 1 s & Update rate of target networks $\nu$ & 0.9 \\
\hline
\end{tabular}
\vspace{-1mm}
\end{table}

In this section, we present the simulation results on the AoI minimization for UAVs in the system. The simulation parameters are based on the existing 3GPP technical reports~\cite{3GPP_TR_36_777,3GPP_TR_38_901}, which are given in Table~\ref{parameters}.

In the simulation, we model the cell as a circular area whose center is the BS and the radius is $500~\textrm{m}$. The sensing targets and mobile devices are randomly distributed within the cell, as referred in~\cite{WZTMA-2009}. The initial positions of UAVs in the U2N mode are on the BS with the altitude of $100~\textrm{m}$, while those of UAVs in the U2D mode are on their corresponding mobile devices with the same altitude. Moreover, we set each actor or critic network as a four-layer neural network with two hidden layers~\cite{YZFZ-2018}, in which the numbers of neurons in the two hidden layers are $200$ and $100$, accordingly. During the training of DNNs, we adopt the rectified linear unit~(ReLU) function, defined by $f_{ReLU}(x) = \textrm{max}\{0,x\}$, as the activation function~\cite{GTG-2013}. The learning rate are set as exponentially decayed to improve the performance of training~\cite{ZPHJ-2012}.

\begin{figure}[!t]
\centering
\vspace{-3mm}
\includegraphics[width=4.2in]{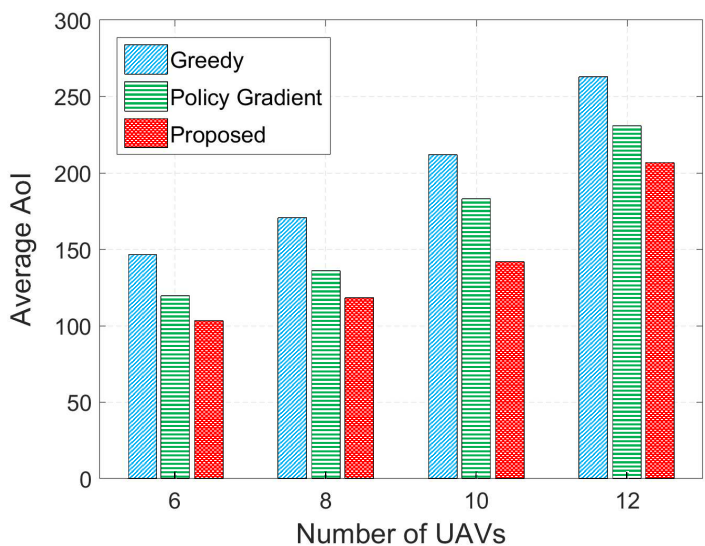}
\vspace{-7mm}
\caption{Performance comparison of different algorithms, given different numbers of UAVs in the system. }
\vspace{-7mm}
\label{Plot_Performance}
\end{figure}

In Fig.~\ref{Plot_Performance}, we compare the performance of our proposed algorithm with the following two algorithms:
\begin{itemize}
  \item \textbf{Greedy algorithm}: Each UAV determines its sensing and transmission locations to maximize its SSP and STP.
  \item \textbf{Policy gradient algorithm}~\cite{RDSY-1999}: Each UAV directly optimizes its parameterized control policy by a variant of gradient descent.
\end{itemize}
For all cases, we assume that the numbers of UAVs in the U2N and the U2D modes are identical. As is shown in the figure, our proposed algorithm can obtain a lower average AoI in the system than the greedy algorithm and the policy gradient algorithm. In addition, for each algorithm, the average AoI increases with the number of UAVs, since more sensing tasks to be completed.

\begin{figure}[!t]
\centering
\vspace{-3mm}
\includegraphics[width=4.2in]{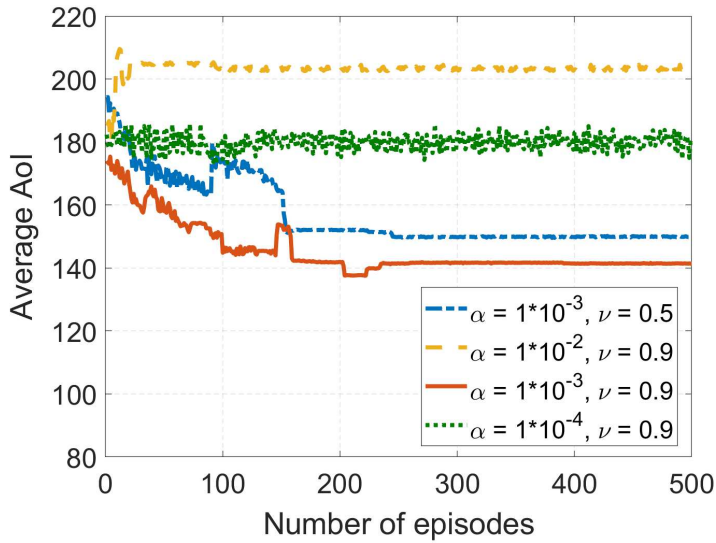}
\vspace{-7mm}
\caption{Convergency of the proposed algorithm, given different learning rates $\alpha$ and update rates of target networks $\nu$.}
\vspace{-7mm}
\label{Plot_Convergency}
\end{figure}

Fig.~\ref{Plot_Convergency} shows the convergency of our proposed algorithm with different learning rates~$\alpha$ and update rates of target networks~$\nu$. Under appropriate hyperparameters, e.g., $\alpha = 1\times10^{-3}$ and $\nu = 0.9$, our proposed algorithm converges after $250$ episodes with satisfied performance. When the learning rate is too large, e.g, $\alpha = 1\times10^{-2}$, the algorithm converges quickly while its performance cannot be guaranteed. On the other hand, when  the learning rate is too small, e.g, $\alpha = 1\times10^{-4}$, the duration of convergence is quite long. Moreover, the algorithm can achieve a lower average AoI when the update rate of target networks is larger.

\begin{figure}[!t]
\centering
\includegraphics[width=6.4in]{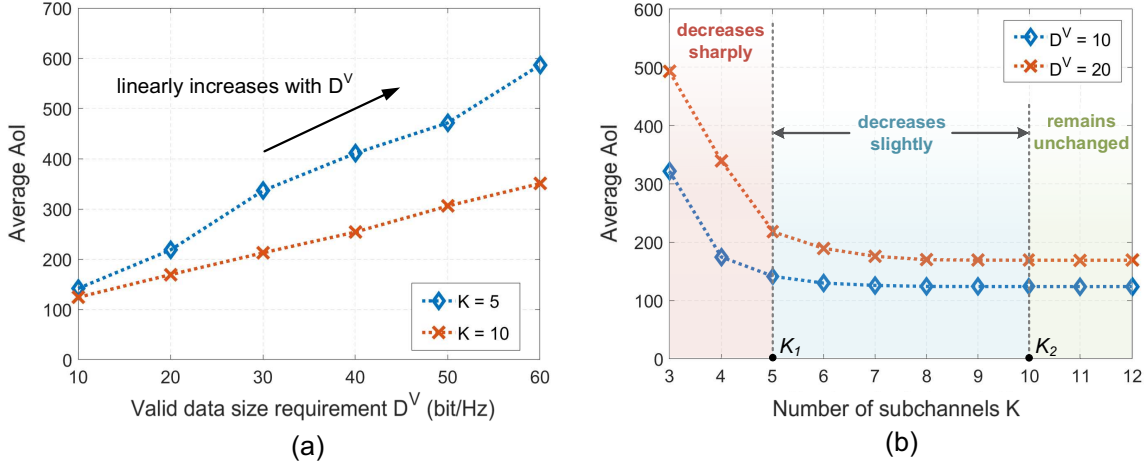}
\vspace{-9mm}
\caption{(a)~The average AoI versus valid data size requirement $D^V$, given different numbers of subchannels $K$;~(b)~The average AoI versus subchannels $K$, given different numbers of valid data size requirement $D^V$, in which $K_1 = \max\{M,{N}/{V_s}\} = 5$ and $K_2 = M+N = 10$. }
\vspace{-7mm}
\label{Plot_Data_Subchannel}
\end{figure}

Fig.~\ref{Plot_Data_Subchannel} presents the effects from UAV sensing and transmission on the average AoI in the system, which justifies the AoI analysis in Section~\ref{System Performance Analysis}. Specifically, in Fig.~\ref{Plot_Data_Subchannel}(a), we plot the average AoI versus valid data size requirement $D^V$, given different numbers of subchannels $K$. For each case, we can observe that the average AoI linearly increases with $D^V$, which is consistent with Proposition~\ref{Remark1}. Besides, a higher value of $K$ leads to a lower average AoI, as more spectrum resource can be utilized. Furthermore, in Fig.~\ref{Plot_Data_Subchannel}(b), we show the average AoI versus the number of subchannels $K$, given different numbers of valid data size requirement $D^V$. In either case,  when $K$ is smaller than $5$, the average AoI is dramatically decreases with $K$. When $K$ increases from $5$ to $10$, the average AoI decreases gently. As $K$ exceeds $10$, the average AoI remains unchanged. This is consistent with Proposition~\ref{Remark2}. In addition, when $D^V$ gets larger, the average AoI becomes higher, since more sensory data needs to be collected.

\begin{figure}[!t]
\centering
\includegraphics[width=6.4in]{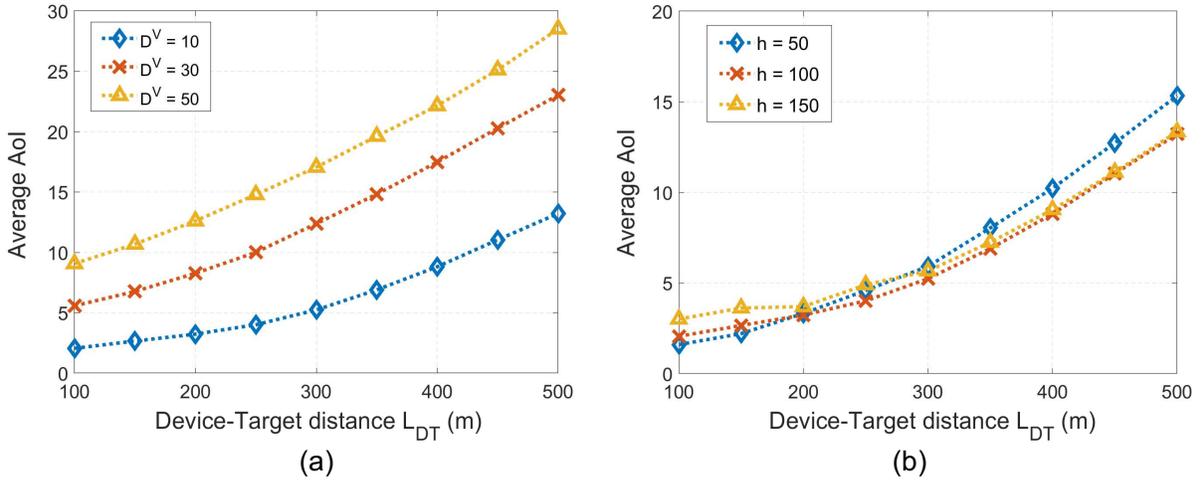}
\vspace{-7mm}
\caption{(a)~The average AoI of a UAV in the U2D mode versus the device-target distance~$L_{DT}$, given different numbers of valid data size requirement $D^V$, with UAV flying altitude $h = 100$~m;~(b)~The average AoI of a UAV in the U2D mode versus the device-target distance~$L_{DT}$, given different numbers of UAV flying altitude $h$, with valid data size requirement $D^V = 10$~bit/Hz.}
\vspace{-7mm}
\label{Plot_Distance}
\end{figure}

In Fig.~\ref{Plot_Distance}, we show the average AoI of a UAV in the U2D mode\footnote{This simulation result is also appropriate for a UAV in the U2N mode if we replace the mobile device by the BS.} versus the distance from its mobile device to its target~$L_{DT}$, with different numbers of valid data size requirement $D^V$ and UAV flying altitude $h$. In the simulation, we assume that the UAV flies along the straight line between its mobile device and its target at the fixed altitude $h$. Besides, we assume that the interference from co-channel UAVs is controllable and can be omitted for simplicity. The average AoI is obtained by simulation, in which we enumerate all possible sensing and transmission locations and find the ones with the lowest AoI. In Fig.~\ref{Plot_Distance}(a), we can find out that the average AoI increases superlinearly with $L_{DT}$, since the time spend on the flight increases. In addition, as the increase of $D^V$,  the curves become more linear. From Fig.~\ref{Plot_Distance}(b), we can also observe that, for any flying altitude, the average AoI increases superlinearly as $L_{DT}$ gets larger. Besides, when the UAV's mobile device is close to its target, it is more appropriate for the UAV to fly at a low altitude. As the device-target distance becomes larger, the UAV tends to properly increase its flying altitude.

\section{Conclusion}
\label{Conclusion}
In this paper, we have proposed the underlaying U2D communications in a cellular Internet of UAVs, and studied the AoI minimization problem in this network. We have designed a joint sensing and transmission protocol to schedule multiple UAVs performing sensing tasks. Since the AoI minimization problem can be regarded as a MDP, we have formulated this problem by multi-agent DRL, and proposed a DDPG-based multi-UAV trajectory design algorithm to solve this problem. Simulation results have shown that our proposed algorithm outperforms the greedy algorithm and the policy gradient algorithm. Two conclusions on the AoI can be drawn from the simulation results. First, the UAV's AoI linearly increases with the sensory data demand. Second, the UAV's AoI decreases with the number of subchannels, and then becomes saturated.

\begin{appendices}
\section{Proof of Proposition 1}\label{proof1}
Given the threshold $R_{th}$, we can express the STP for the $i$-th UAV~$( i \in \mathcal {M} )$ over the $k$-th subchannel as
\begin{align}\label{STP_derivation}
\mathcal{P}^{st}_{i,k} \left( R_{th} \right) &= \mathcal{P}\left\{ {{\log }_2}\left( {1 + {\gamma _{i,k}}} \right) > R_{th} \right\} \nonumber \\
&= \mathcal{P}\left\{ {\gamma _{i,k}} > {2^{R_{th} }} - 1\right\}  \nonumber \\
&\mathop = \limits^{(a)}  \mathcal{P}\left\{ \zeta_i  > \frac{N_0\kappa }{P^u g_i } + \sum\limits_{j \in \mathcal {N}}\phi_{j,k} \frac{ \kappa }{ g_i }g_j\zeta_j \right\} \nonumber \\
&\mathop = \limits^{(b)}  \mathcal{P}^{LoS}_i \cdot [1-F_{ri}(\chi)]+ \mathcal{P}^{NLoS}_i \cdot [1-F_{ra}(\chi)].
\end{align}
Here, equation~(a) holds because the subchannel allocation indicator $\phi_{i,k}$ equals to one if the $i$-th UAV is assigned to the $k$-th subchannel, and we define $\kappa = 2^{R_{th}}-1$ for simplicity. Besides, equation~(b) is due to that $\zeta_i$ obeys Rice distribution when the LoS component exists, and follows Rayleigh distribution when the NLoS component exists. To be specific, $F_{ri}(\chi) = 1 - Q_1(\sqrt{2K_{ri}},\chi \sqrt{2(K_{ri}+1)})$ denotes the cumulative distribution function~(CDF) of the Rice distribution with $\Omega = 1$~\cite{S-1944}, $F_{ra}(\chi) = 1 - e^{-\chi^2/2}$ denotes the CDF of the Rayleigh distribution with unit variance, and $Q_1(x)$ is the Marcum Q-function of order 1~\cite{J-1950}. We define $\chi = A_i + \sum_{j \in \mathcal {N}} \phi_{j,k} B_i \chi_j$, in which $ A_i = \frac{N_0\kappa }{P^u g_i }$, $B_i = \frac{ \kappa }{g_i}$, and $\chi_j = g_j \zeta_j$~$(j\in\mathcal {N})$. Note that $A_i$ and $B_i$ can be regarded as constants if the location of the $i$-th UAV is given. Therefore, the PDF of $\chi$, denoted by $f_\chi(x)$, can be derived from the PDFs of $\chi_j$.

Due to the mutual independence among small-scale fading coefficients, $\chi_j$~$(j\in\mathcal {N})$ are independent with each other. Thus, the PDF of $\chi$ can be given by $f_\chi(x) = g(x-A_i)$, where $ g(x) = g_{\mathcal {W}[1]}(x) \ast \ldots \ast g_{\mathcal {W} [N_w]}(x)$. Here, $\mathcal {W} = \left\{ j \in \mathcal {N} | \phi_{j,k}=1  \right\}$ is a set with size $N_w$, where $\mathcal {W}[w]$ denote the $w$-th element in $\mathcal {W}$~$(w=1,...,N_w)$. Moreover,  $g_j(x)= ({1}/{B_i})  f_{\chi_j}({x}/{B_i})$,~$(j \in \mathcal {W})$, in which $f_{\chi_j}(y)= \mathcal{P}^{LoS}_j \cdot [({1}/{g_j}) f_{ri}(y/{g_j})] + \mathcal{P}^{NLoS}_j \cdot [({1}/{g_j}) f_{ra}(y/{g_j})]$, with $f_{ri}(y)= 2\left(K_{ri}+1\right)y e^{ -\left(K_{ri}+1\right)y^2-K_{ri} } \cdot I_0\left(2\sqrt{(K_{ri}+1)K_{ri}} y\right)$ and $f_{ra}(y)=y e^{-y^2/2}$. Note that $f_\chi(x)$ is the multiple convolutions of Rice and Rayleigh PDFs, which is quite complex. As such, we can hardly derive the close-form expression of $ f_\chi(x)$. Instead, we will obtain the numerical result of $f_\chi(x)$ by simulation. Finally, we can calculate the expected value of $\mathcal{P}^{st}_{i,k} \left( R_{th} \right)$ based on $f_\chi(x)$, which ends the proof.

\section{Proof of Proposition 2}\label{proof2}
As the data transmission will fail if the throughput is lower than the given threshold $R_{th}$, the expected throughput for the $i$-th UAV over the $k$-th subchannel, defined as $ER_{i,k} = \mathbb{E} \{ R_{i,k} \}$, can be calculated by
\begin{align}
ER_{i,k} &= \int_{R_{th}}^\infty {  r \cdot f_R(r)  } dr \nonumber \\
&= \int_{R_{th}}^\infty {  r  } d[ F_R(r) ] \nonumber \\
&=  \left. {r \cdot F_R(r) }\right|_{R_{th}}^\infty  - \int_{R_{th}}^\infty  { F_R(r)  } dr  \nonumber \\
&=  \left. {r \cdot F_R(r) } \right|_0^\infty  -  R_{th} \cdot F_R(R_{th})  - \int_{R_{th}}^\infty  { F_R(r)  } dr  \nonumber \\ 
&\mathop = \limits^{(a)}   \int_{R_{th}}^\infty  [ 1-F_R(r) ] dr  + R_{th} \cdot [ 1 - F_R(R_{th}) ] \nonumber \\
&\mathop = \limits^{(b)}  \int_{R_{th}}^\infty \mathcal{P}^{st}_{i,k} \left( r \right) dr  +  R_{th} \cdot \mathcal{P}^{st}_{i,k} \left( R_{th} \right)
\end{align}
in which  $f_R(r)$ and $F_R(r)$ are defined as the PDF and the CDF of $R_{i,k}$, respectively. Here, equation~(a) can be obtained from $\left. {r \cdot F_R(r) } \right|_0^\infty = \left. r \right|_0^\infty = \int_0^{\infty} { 1  } dr = \int_0^{R_{th}}  { 1  } dr  + \int_{R_{th}}^{\infty}  { 1  } dr$. Besides,  equation~(b) holds due to that $1-F_R(r) = \mathcal {P}\{ R_{i,k}>r \} = \mathcal{P}^{st}_{i,k}(r)$. Then, the proof ends.

\end{appendices}


\end{document}